\documentclass[a4paper]{article}
\usepackage{fullpage}
\usepackage{amssymb, amsmath, amsthm}
\usepackage{graphicx}
\graphicspath{{figures/}}
\usepackage[utf8]{inputenc}
\usepackage[T1]{fontenc}
\usepackage{xcolor}
\usepackage{xspace}
\usepackage{framed}
\usepackage{enumitem}
\usepackage{listings}
\usepackage[numbers,sort&compress]{natbib}
\usepackage{authblk}
\usepackage{subcaption}
\usepackage{hyperref}
\usepackage{thmtools}
\usepackage{thm-restate}

\theoremstyle{theorem}
\newtheorem{theorem}{Theorem}
\newtheorem{lemma}[theorem]{Lemma}

\newtheorem{observation}{Observation}

\newcommand{\poly}{\ensuremath{\mathcal{P}}\xspace} %polygon
\def\C{\ensuremath{\mathcal{C}}\xspace}
\newcommand{\R}{\ensuremath{\mathbb{R}}\xspace}
\newcommand{\Q}{\ensuremath{\mathbb{Q}}\xspace}
\newcommand{\NP}{\ensuremath{\text{NP}}\xspace}
\newcommand{\APX}{\ensuremath{\text{APX}}\xspace}
\newcommand{\etal}{\emph{et al.}}

\newcommand{\comm}[1]{}

\title{Irrational Guards are Sometimes Needed}
\author[1]{Mikkel Abrahamsen}

\author[1]{Anna Adamaszek}

\author[2]{Tillmann Miltzow}

\affil[1]{University of Copenhagen, Denmark.
\texttt{\{miab,anad\}@di.ku.dk}}		
\affil[2]{Institute for Computer Science and Control, 
Hungarian Academy of Sciences (MTA SZTAKI), \texttt{t.miltzow@gmail.com}}

\date{}

\begin{document}
\maketitle

\vfill

% \begin{figure}[h]
  {\centering
  \includegraphics[width=.77\textwidth]{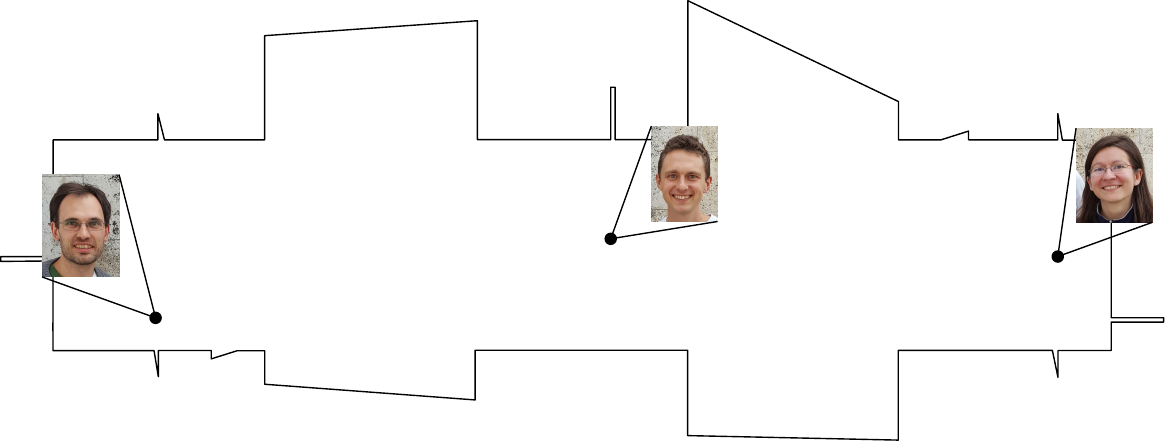}
  
  Till, Mikkel, Anna meticulously guarding the polygon: a little irrational, but pretty optimal.}
% \end{figure}

\vfill

\begin{abstract}
In this paper we study the \emph{art gallery problem}, which is one of the fundamental problems in computational geometry. The objective is to place a minimum number of guards inside a simple polygon such that the guards together can see the whole polygon. 
We say that a guard at position $x$ sees a point $y$ if the line segment $xy$ is fully contained in the polygon.

Despite an extensive study of the art gallery problem, it remained an open question whether there are polygons given by integer coordinates that require guard positions with irrational coordinates in any optimal solution. We give a positive answer to this question by constructing a \emph{monotone} polygon with integer coordinates that can be guarded by three guards only when we allow to place the guards at points with irrational coordinates. Otherwise,
four guards are needed.
By extending this example, we show that for every $n$,
there is polygon which can be guarded by $3n$ guards with irrational coordinates
but need $4n$ guards if the coordinates have to be rational.
Subsequently, we show that there are
rectilinear polygons given by integer coordinates that
require guards with irrational coordinates in any optimal solution.
\end{abstract}
\vfill

\newpage

%----------------------------------------------------------------------------------------------
\section{Introduction}

For a polygon $\poly$ and points $x,y \in \poly$, we say that $x$ \emph{sees} $y$ if the interval $xy$ is contained in $\poly$.
  A \emph{guard set} $S$ is a set of points in $\poly$ such that
  every point in $\poly$ is seen by some point in $S$.
  The points in $S$ are called \emph{guards}.
  The \emph{art gallery problem} is to find 
  a minimum cardinality guard set for a simple polygon $\poly$ on $n$ vertices.
  The polygon \poly is considered to be filled,
  i.e., it consists of
  a closed polygonal curve in the plane and the bounded
  region enclosed by this curve.

  This classical version of the art gallery problem
  has been originally formulated in 1973 by Victor Klee (see the book of
  O'Rourke~\cite[page 2]{o1987art}).
  It is often referred to as 
  the \emph{interior-guard art gallery problem} or the \emph{point-guard art gallery problem}, to distinguish it from other versions that have been introduced over the years. 
  
In 1978, Steve Fisk provided an elegant proof that $\lfloor n/3 \rfloor$ guards are always sufficient and sometimes necessary to guard a polygon with $n$ vertices~\cite{DBLP:journals/jct/Fisk78a}.
Five years earlier, Victor Klee had posed this question to V\'aclav Chv\'atal, who soon gave a more complicated solution~\cite{chvatal1975combinatorial}. Since then, the art gallery problem has been extensively studied, both from the combinatorial and the algorithmic perspective.
 Most of this research, however, is not focused directly on the classical art gallery problem, but on its numerous versions, including different definitions of visibility, restricted classes of polygons, different shapes of the guards,
 restrictions on the positions of the guards, etc. For more detailed information we refer the reader to the following surveys~\cite{shermer1992recent, urrutia2000art, o1987art, ghosh2007visibility}.
  
%Finding an optimal guard set for the art gallery problem is known to be NP-hard
%\cite[pp.\ 239--242]{o1987art}. 
Despite extensive research on the art gallery problem, no combinatorial algorithm for finding an optimal solution, or even for deciding whether a guard set of a given size $k$ exists, is known.
  The only exact algorithm is
  attributed to Micha Sharir (see~\cite{DBLP:journals/ipl/EfratH06}), who
  has shown that
  in $n^{O(k)}$ time one can find a guard set consisting of $k$ guards, if such a guard set exists.
  This result is obtained by using standard tools from real algebraic geometry~\cite{basu2006algorithms}, and it is not known how to find an optimal solution without using this powerful machinery (see~\cite{belleville1991computing} for an analysis of the very restricted case of $k=2$). 
  To stress this even more: Without the tools from algebraic geometry,
  we would not know if it is decidable whether a guard set of
  size $k$ exists or not!
  Some recent lower bounds~\cite{DBLP:conf/esa/BonnetM16} based on 
  the exponential time hypothesis
  suggest that there might be no better exact algorithms than the
  one by Sharir. 

  To explain the difficulty in constructing exact algorithms, we want to emphasize that it is \emph{not} known whether the decision version of the art gallery problem (i.e., the problem of deciding whether there is a guard set consisting of $k$ guards, where $k$ is a parameter) lies in the complexity class \NP, even with the algorithm by Sharir. 
  While $\NP$-hardness and $\APX$-hardness of the art gallery problem have been shown for different classes of polygons~\cite{DBLP:journals/tit/LeeL86, DBLP:journals/mlq/SchuchardtH95,tomas2013guarding, broden2001guarding, eidenbenz2001inapproximability,o1983some, MonotonHard},
  the question of whether the point-guard art gallery problem is in $\NP$ remains open.
  A simple way to show $\NP$-membership would be to prove that there always exists an optimal set of guards with \emph{rational} coordinates of polynomially bounded description.

  Indeed,
  S\'{a}ndor Fekete posed at MIT in 2010 and at Dagstuhl in 2011 an open problem, asking
  whether there are polygons requiring
  irrational coordinates in an optimal guard set~\cite{SandorEmails,dagstuhlSeminar}.
  The question has been raised again by G\"{u}nter Rote at EuroCG 2011~\cite{OpenProblem}.
  It has also been mentioned by Rezende \etal~\cite{RezendeSFHKT14}: ``it remains an open question whether there are polygons
given by rational coordinates that require optimal guard positions with irrational
coordinates''.
A similar question has been raised by Friedrichs \etal~\cite{DiscretizeTerrain}:
``[\ldots] it is a long-standing open problem for the more
general Art Gallery Problem (AGP): For the AGP it is not known whether the coordinates
of an optimal guard cover can be represented with a polynomial number of bits''. 

  \medskip
  \noindent \textbf{Our results.}
  We answer the open question of S\'{a}ndor Fekete, by proving the following main
  result of our paper. Recall that a polygon $\poly$ is called \emph{monotone} if there exists a line $l$ such that every line orthogonal to $l$ intersects $\poly$ at most twice.

  \begin{theorem}\label{thm:3-irrational}
    There is a simple monotone polygon $\poly$ with integer coordinates of the vertices such that
  \begin{enumerate}[label=(\roman*), itemsep = 0cm]
    \item  $\poly$ can be guarded by $3$ guards placed at points with irrational coordinates, and
    \item  an optimal guard set of $\poly$ with guards at points with rational coordinates has size $4$.
  \end{enumerate}
\end{theorem}

We then extend this result, by providing a family of polygons for which the ratio between the number of guards in an optimal solution restricted to guards at rational positions, to the number of guards in an optimal solution allowing irrational guards, is $4/3$. 

    \begin{theorem}\label{thm:4-3-bound}
    There is a family of simple polygons $(\poly_n)_{n \in \mathbb{Z}_+}$ with integer coordinates of the vertices such that 
    \begin{enumerate}[label=(\roman*), itemsep = 0cm]
      \item $\poly_n$ can be guarded by $3n$ guards placed at points with irrational coordinates, and
     \item an optimal guard set of $\poly_n$ with guards at points with rational coordinates has size $4n$.
    \end{enumerate} 
Moreover, the coordinates of the points defining the polygons $\poly_n$ are polynomial in $n$.
\end{theorem}
  
We show that the phenomenon with guards at irrational coordinates occurs also in the important class of rectilinear polygons.
  
\begin{restatable}%[Rectilinear Polygon]
    {theorem}{rectilinear}
    \label{thm:rectilinear}
There is a rectilinear polygon $\poly_R$ with vertices at integer coordinates satisfying the following properties.
    \begin{enumerate}[label=(\roman*), itemsep = 0cm]
      \item \label{itm:9guards} $\poly_R$ can be guarded by $9$ guards if we allow placing guards at points with irrational coordinates.     \item  \label{itm:10guards} An optimal guard set of  $\poly_R$ with guards at points with rational coordinates has size $10$.    \end{enumerate}
\end{restatable}
    
  \medskip
  \noindent \textbf{Other related work.}
  The art gallery problem has been studied from the perspective of approximation algorithms.
   %For general polygons, Deshpande et al.\ gave a randomized pseudo-polynomial time $O(\log n)$-approximation algorithm~\cite{DeshpandeThesis, DeshpandeWADS2007}. However, Bonnet and Miltzow
   %showed that their algorithm is not correct~\cite{DBLP:journals/corr/BonnetM16b}.  
   %Some approximation algorithms for the point-guard art gallery problem
   %have been described.
   Efrat and Har-Peled~\cite{DBLP:journals/ipl/EfratH06} gave a randomized  
   polynomial time
   algorithm 
   for finding a guard set $S$ where the guards are
   restricted to a very fine grid $\Gamma$. 
   To be more precise, if coordinates of the vertices of the input polygon $\poly$ are given by
   positive integers and $L$ is the largest such integer,
   then $\Gamma$ can be defined as the points in
   $L^{-20}\cdot \mathbb{Z}^2 \cap \poly$.
   Let $OPT_{\text{grid}}\subset \Gamma$ be a guard
   set with a minimum number of guards under this restriction.
   The algorithm of Efrat and Har-Peled yields an $O(\log |OPT_{\text{grid}}|)$-approximation for this problem.
   However, it remained
   open whether $OPT_{\text{grid}}$ is an
   approximation of an optimal unrestricted guard set $OPT$.
  Bonnet and Miltzow~\cite{DBLP:journals/corr/BonnetM16b}
  filled this gap by showing that under a general position
  assumption
  $|OPT_{\text{grid}}|= O(|OPT|)$,
  which yields the first polynomial time approximation algorithm for simple polygons under this assumption.
  It is easy to construct a polygon with integer coordinates that forces 
  a guard on the point $(1/3,1/3)$, which might not lie on the grid, in case that $L$ is not divisible by $3$. This implies that $OPT_{\text{grid}}$ is not optimal. But this does not rule out that there is a slightly more clever choice of $\Gamma$ so that $OPT_{\text{grid}}$ is indeed optimal.
  It follows from our
  Theorem~\ref{thm:4-3-bound} that
  there are polygons (requiring arbitrarily many guards in an optimal guard set) such that 
  for any choice of $\Gamma\subset\Q^2$, 
   it holds that $|OPT_{\text{grid}}|\geq 4/3\cdot |OPT|$.
  No lower bound of this kind has been known before.
  More general, our result shows that no algorithm which considers only rational points as possible guard positions can achieve an approximation ratio better than $4/3$.

 A new line of research focuses on implementing algorithms that are capable 
 of solving instances of the art gallery problem with thousands of vertices, giving a solution which is close to the optimal one, see the recent survey by Rezende \etal~\cite{RezendeSFHKT14}.
 They explain that many
 practical algorithms rely on
 ``routines to find candidates for discrete guard and witness locations.'' 
 We show that this technique inevitably leads to sub-optimal solutions unless irrational
 candidate locations are also considered. 
    We believe that our example and techniques are a good starting point to construct benchmark instances for implementations of art gallery algorithms.
  Benchmark instances serve to validate the quality of algorithms. Using the same instances when comparing different algorithms makes the results comparable.

 A problem related to the art-gallery problem is the \emph{terrain guarding problem}.
  In this problem, an $x$-monotone polygonal curve $c$ (i.e., the terrain) is given.
  The region $R$ above the curve $c$ has to be guarded,
  and the guards are restricted to lie on $c$.
  Similarly as in our problem, a guard $x$ sees a point $y$ if $xy$ is contained in the region $R$.
  Although the solution space of the terrain guarding instance is the continuous polygonal curve $c$,
  a discretization of the solution space has been recently described by
Friedrichs \etal~\cite{DiscretizeTerrain}.
Given a terrain with $n$ vertices at integer position, they 
describe a set $S\subset \Q$ of size $O(n^3)$, computable in polynomial time,
such that there is an optimal guard placement restricted to $S$.
It follows that for the terrain guarding problem the phenomenon with irrational numbers does not appear, and also the decision version of the terrain guarding problem is in $\NP$.

%\mikkel{Another related problem that we could mention (possibly instead of
%terrain guarding) is guarding orthogonal polygons with holes using the
%rectangular visibility. That is NP-hard, but also NP-complete, since
%there is an optimal guard set where the guards are placed on
%intersections of extensions of the line segments. I have asked Frank Hoffmann
%for a reference.}

Irrational numbers turn up surprisingly in other areas of computational geometry. One such example is the
  \emph{nested polytopes problem}. Here, we are given two nested
   polytopes $S\subseteq P$ and want to find a polytope $T$ with a minimum
   number of corners such that $T$ is nested between $S$ and $P$,
   i.e., $S\subseteq T\subseteq P$.
  Christikov \etal~\cite{chistikov_et_al:LIPIcs:2016:6238} recently
  gave an example of two
  nested polytopes $S\subseteq P$ in $\R^3$, with all corners at rational
  coordinates, such that there is a unique
  polytope $T$ with $5$ corners nested between $S$ and $P$, and $T$ has corners
  with irrational coordinates.
  The nested polytopes problem is closely related to \emph{nonnegative matrix
  multiplication},
  where similar phenomena have been discovered, that
  a problem defined entirely by rational numbers has an optimal
  solution requiring irrational numbers~\cite{chistikov_et_al:LIPIcs:2016:6238,
  chistikov2016nonnegative}.

  \medskip
  \noindent \textbf{The Structure of the Paper.}
  Section~\ref{sec:polygon} contains the description of
  a monotone polygon $\poly$ with vertices at points with rational coordinates that can be guarded
  by three guards only if the guards are placed at points with irrational coordinates.
  In Section~\ref{sec:intuition}, we describe the intuition behind our
  construction, and explain how we have found the polygon $\poly$.
  The formal proof of Theorems \ref{thm:3-irrational} and \ref{thm:4-3-bound}
  %that $\poly$ requires irrational guards
  is then provided
  in Section~\ref{proofSec}.
  In Section~\ref{sec:rectilinear}, we present
  the rectilinear polygon $\poly_R$ from Theorem~\ref{thm:rectilinear}
  requiring guards with irrational coordinates in an optimal guard set.
  Finally, in Section~\ref{sec:future} we suggest some open problems for future research.

\section{The Polygon}\label{sec:polygon}

In Figure \ref{fig:polygon} we present the polygon $\poly$. In Section~\ref{proofSec} we will prove that $\poly$ can be guarded by three guards only when we allow the guards to be placed at points with irrational coordinates.

\begin{figure}
\centering
\includegraphics[width=1.55\textwidth, angle =270]{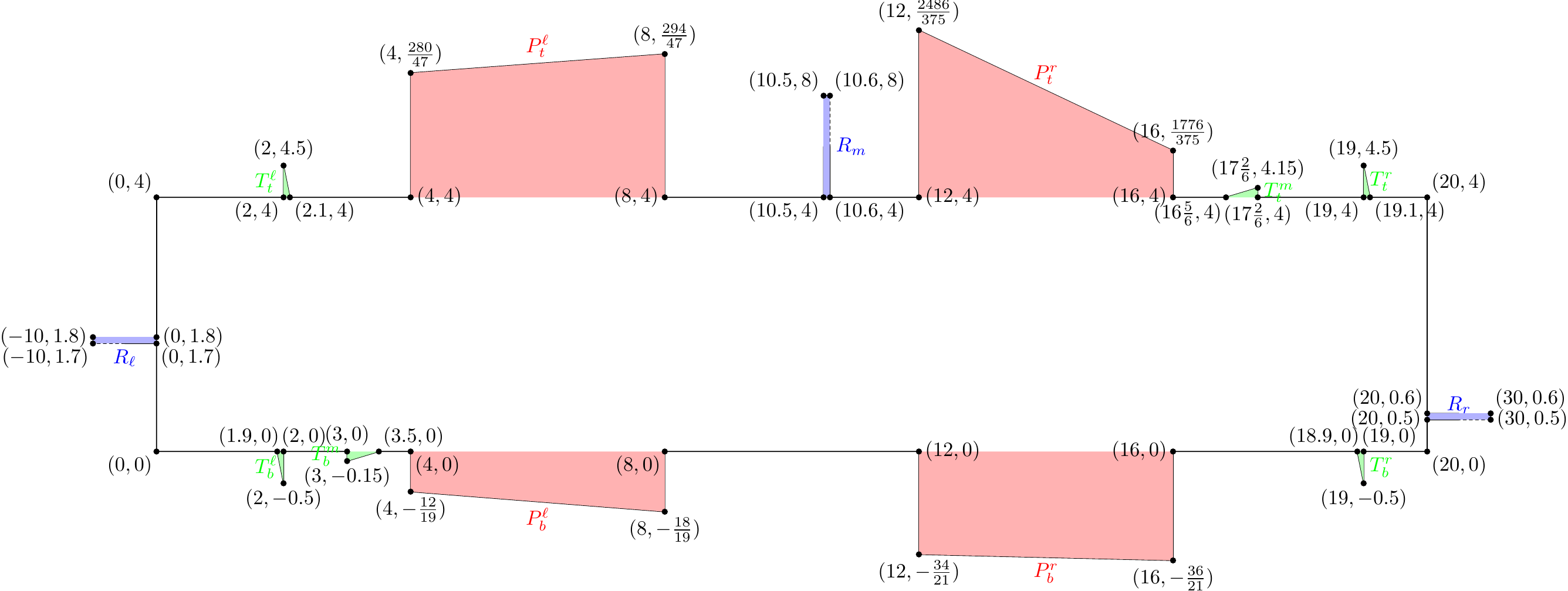}
\caption{The polygon $\poly$. We will show that $\poly$ can be guarded by three guards only when we allow the guards to be placed at points with irrational coordinates.
 For practical reasons, the blue rectagular pockets are drawn shorter than they actually are.}
\label{fig:polygon}
\end{figure}

The polygon $\poly$ is constructed as follows. We start with a \emph{basic rectangle} $[0,20]\times[0,4] \subset \mathbb{R}^2$. Then, we append to it six \emph{triangular pockets} (colored with green in the figure), which are triangles defined by the following coordinates

\begin{tabular}{lcl}
$T^\ell_t:\ \{(2,4),(2,4.5),(2.1,4)\}$, & & $T^\ell_b:\ \{(2,0),(2,-0.5),(1.9,0)\}$,\\
$T^m_t:\ \{(16\frac{5}{6},4),(17\frac{2}{6},4.15),(17\frac{2}{6},4)\}$, & & $T^m_b:\ \{(3.5,0),(3,-0.15),(3,0)\}$,\\ 
$T^r_t:\ \{(19,4),(19,4.5),(19.1,4)\}$, & and &  $T^r_b:\ \{(19,0),(19,-0.5),(18.9,0)\}$.\\
\end{tabular}

Next, we append three \emph{rectangular pockets} (colored with blue in the figure,
for practical reasons these pockets are drawn in the figure shorter than they actually are),
which are rectangles defined in the following way.\\
\indent $R_\ell$: $[-10,0]\times [1.7,1.8]$, $R_r$: $[20,30]\times [0.5,0.6]$,
and $R_m$:  $[10.5,10.6]\times [4,8]$.

Last, we append four \emph{quadrilateral pockets} (colored with red in the figure), which are defined by points with the following coordinates

\begin{tabular}{lllll}
Top-left pocket $P^\ell_t$ & $\{(4,4),$ & $(4,\frac{280}{47}),$ & $(8,\frac{294}{47}),$ & $(8,4)\}$\\
Top-right pocket $P^r_t$ & $\{(12,4),$ & $(12,\frac{2486}{375}),$ & $(16,\frac{1776}{375}),$ & $(16,4)\}$\\
Bottom-left pocket $P^\ell_b$ & $\{(4,0),$ & $(4,-\frac{12}{19}),$ & $(8,-\frac{18}{19}),$ & $(8,0)\}$\\
Bottom-right pocket $P^r_b$ & $\{(12,0),$ & $(12,-\frac{34}{21}),$ & $(16,-\frac{36}{21}),$ & $(16,0)\}$.\\
\end{tabular}

The polygon $\poly$ is clearly monotone.
We will denote by $e^\ell_t$, $e^r_t$, $e^\ell_b$, and $e^r_b$ the non-axis-parallel
edge within each of the four quadrilateral pockets, respectively.

\section{Intuition}\label{sec:intuition}

In this section, we explain
the key ideas behind the construction of the polygon $\poly$.
Our presentation is informal, but it
resembles the work process that lead to the construction of $\poly$ more than the
formal proof of Theorem~\ref{thm:3-irrational} in Section~\ref{proofSec} does. 
Here we omit all ``scary'' computations and focus on conveying the big picture. In the end of this section, we also explain how we actually constructed the polygon $\poly$.

Define a \emph{rational point} to be a point with two rational coordinates.
An \emph{irrational point} is a point that is not rational.
A \emph{rational line} is a line that contains two rational points.
An \emph{irrational line} is a line that is not rational.

  \medskip
\noindent \textbf{Forcing a Guard on a Line Segment.}
Consider the drawing of the polygon $\poly$ in Figure \ref{fig:polygon}. We will now explain an idea of how three pairs of triangular pockets, $(T^\ell_t, T^\ell_b)$, $(T^m_t, T^m_b)$, and $(T^r_t, T^r_b)$, can enforce three guards on three line segments within $\poly$.

Consider the two triangular pockets in Figure~\ref{fig:ForceLocal}. The blue line segment contains one edge of each of these pockets, and the interiors of the pockets are at different sides of the line segment. A guard which sees the point $t$ must be placed within the orange triangular region, and a guard which sees $b$ must be placed within the yellow triangular region.
Thus, a single guard can see both $t$ and $b$ only if it is on the blue line segment $tb$, which is the intersection of the two regions.

  \begin{figure}[htbp]
  \centering
  \begin{minipage}{0.4\textwidth}
    \centering
    \includegraphics{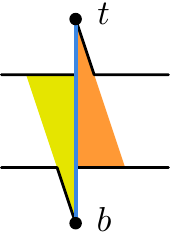}
    \subcaption{The only way that one guard can
    see both $t$ and $b$ is when
    the guard is on the blue line.}
    \label{fig:ForceLocal}
  \end{minipage}
  \hfill
  \begin{minipage}{0.55\textwidth}
    \centering
    \includegraphics{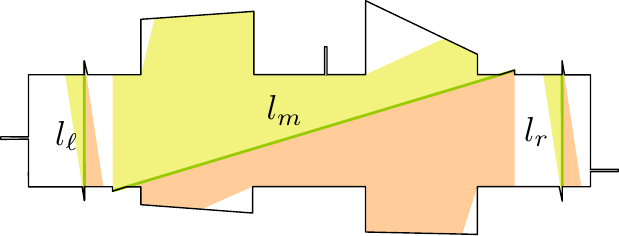}
    \subcaption{The only way to guard the polygon with three guards requires one guard on each of the green lines $l_\ell,l_m,l_r$.}
    \label{fig:ForceGlobal}
  \end{minipage}
  
  \caption{Forcing guards to lie on specific line segments.}
  \label{fig:SimpleLineForcing}
\end{figure}

Consider now the case that we have $k$ pairs of triangular pockets, and no
two regions corresponding to different pairs of pockets intersect.
In order to guard the polygon with $k$ guards,
there must be one guard on the line segment corresponding to each pair.
Our polygon $\poly$ has three such pairs of pockets 
(see Figure~\ref{fig:ForceGlobal}),
and it can be checked that the corresponding regions do not intersect.  

Notice that in this way we can only enforce a guard to be on a rational line, as the line contains vertices of the polygon, which are rational points.
  
  \medskip
\noindent \textbf{Restricting a Guard to a Region Bounded by a Curve.}
For the following discussion, see the Figure~\ref{fig:ConvexRegion} and notation therein. 
We want to guard the polygon from  Figure~\ref{fig:ConvexRegion} using two guards, $g_1$ and $g_2$.
We assume that $g_1$ is forced to be on the blue vertical line $l$. 

Consider some position of $g_1$ on $l$, such that $g_1$ can see at least one point of the top edge $e_t$ of the top quadrilateral pocket, and at least one point of the bottom edge $e_b$ of the bottom quadrilateral pocket. Let $p_t$ and $p_b$ denote the leftmost points
seen by $g_1$ on $e_t$ and $e_b$, respectively.
Observe that $p_t$ moves to the right if $g_1$ moves up, and
to the left if $g_1$ moves down.
The point $p_b$ behaves in the opposite way when $g_1$ is moved.
Consider some fixed position of $g_1$ on the blue line, and the corresponding positions of $p_t$ and $p_b$.
Let $b$ be the bottom right corner of the top pocket, and $d$ the top right corner of the bottom pocket.
Let $i$ be the intersection point of the line containing $p_t$ and $b$, with the line
containing $p_b$ and $d$. 
The points $b,d,i$ define a triangular region $\Delta$.
It is clear that if we place the guard $g_2$ anywhere inside $\Delta$,
then $g_1$ and $g_2$ will together see the entire polygon.
On the other hand, if we place $g_2$ to the right of $\Delta$, then
$g_1$ and $g_2$ will not see the entire polygon,
as some part of the top or the bottom pocket will not be seen.

\begin{figure}[htbp]
  \centering
  \includegraphics{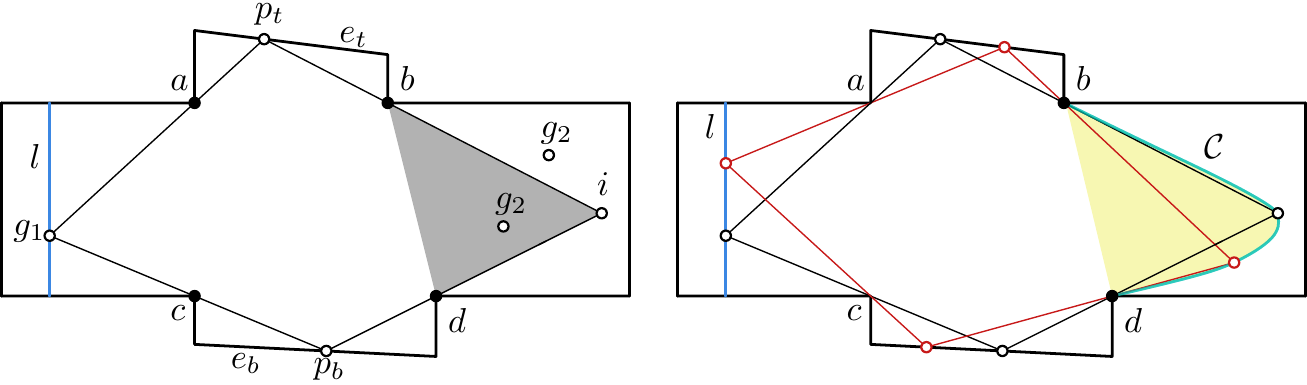}
  \caption{Left: The guard $g_2$ must be inside the triangular region (or to the left of it) in order to guard the entire part of the polygon that is not seen by $g_1$.
  Right: All possible positions of the point $i$ define a simple curve \C.
  }
  \label{fig:ConvexRegion}
\end{figure}

Now, let us move the guard $g_1$ along the blue line.
Each position of $g_1$ yields some intersection point $i$.
We denote the union of all these intersection points by $\C$ (see the right picture in Figure~\ref{fig:ConvexRegion}).
It is easy to see that \C is a simple curve.
We can compute a parameterization of \C since
 we have described how to construct the point $i$ as a function of
 the position of $g_1$.
 
Note that $g_2$ sees a larger part of \emph{both} pockets if it is moved
horizontally to the left and
a smaller part of \emph{both} pockets if it is moved horizontally to the right.
 Consider a fixed position of
  $g_2$ on or to the right of the segment $bd$.
  Let $g_2'$ be the horizontal projection of $g_2$ on \C.
  Let $g_1$ be the unique position on the blue line such that
  $g_1$ and $g_2'$ see all of the polygon.
  If $g_2$ is to the left of \C, $g_2'$ sees less of the pockets than $g_2$,
  so $g_1$ and $g_2$ can together see everything.
  If $g_2$ is to the right of \C, $g_2$ sees less of the pockets than
  $g_2'$ and 
  neither the top nor the bottom pocket are completely guarded by
  $g_1$ and $g_2$. For any higher placement of $g_1$ even less of the top pocket is guarded and for any lower placement of $g_1$ even less of the bottom pocket is guarded. Thus, there exists no placement of $g_1$ such that both pockets are completely guarded by $g_1$ and $g_2$.
  We summarize our reasoning in the following observation.
 
 \begin{observation}\label{curveObs}
  Consider a fixed position of
  $g_2$ on or to the right of the segment $bd$.
  There exists a position of
  $g_1$ on $l$ such that the entire polygon is seen by $g_1$ and $g_2$
  if and only if $g_2$ lies on or to the left of the curve \C.
  \end{observation}
    
  \begin{figure}[htbp]
    \centering
    \includegraphics{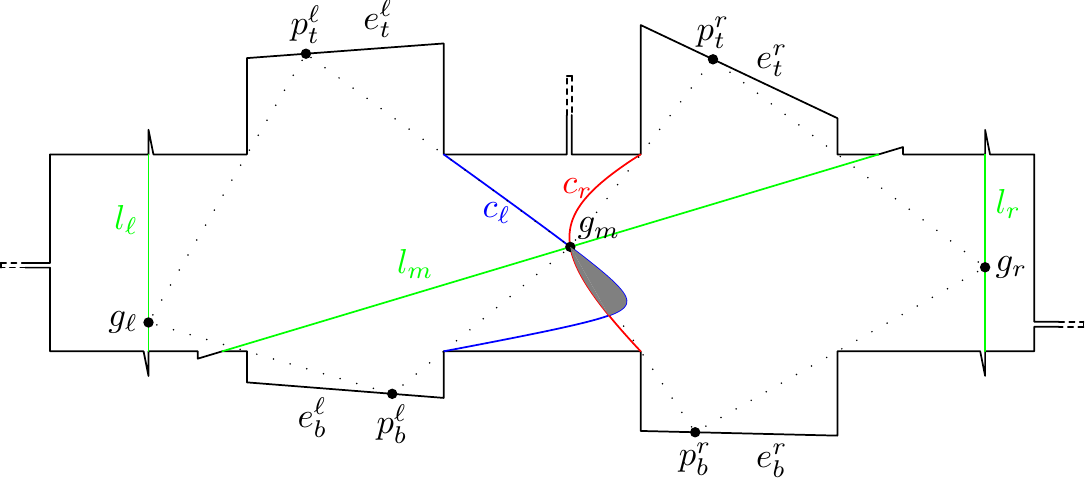}
    \caption{The polygon $\poly$.}
    \label{fig:FullPolygon}
  \end{figure}

    \medskip
  \noindent \textbf{Restricting a Guard to a Single (Irrational) Point.}
  For this paragraph, let us consider the polygon $\poly$ introduced in Section \ref{sec:polygon}, and consider a guard set for $\poly$ consisting of three guards. The polygon $\poly$ is drawn again in Figure~\ref{fig:FullPolygon}, 
  together with additional labels and information.
  The three guards $g_\ell,g_m,g_r$ are forced by the triangular
  pockets to lie on the three 
  green lines $l_\ell, l_m, l_r$, respectively. 
  Additionally, 
  the three rectangular pockets $R_\ell, R_m, R_r$ force the guards to lie 
  within one of three short intervals within each line. 
  (These properties of our construction
  will be discussed in more detail in Section~\ref{proofSec}.)
  With these restrictions, we will show that for the three guards to see the whole polygon, it must hold that the guards $g_\ell$ and $g_m$ can together see the left pockets $P^\ell_t$ and $P^\ell_b$, and the guards $g_m$ and $g_r$ can together see the right pockets $P^r_t$ and $P^r_b$.
  
  Then, the curve $c_\ell$ bounds from the right the feasible region for the guard $g_m$, such that $g_\ell$ and $g_m$ can together see the left pockets $P^\ell_t$ and $P^\ell_b$. Similarly, the curve $c_r$ bounds from the left the feasible region for the guard $g_m$, such that $g_r$ and $g_m$ can together see the right pockets $P^r_t$ and $P^r_b$. Thus, the only way that $g_\ell, g_m$, and $g_r$ can see the whole polygon is when $g_m$ is within the grey region, between $c_\ell$ and $c_r$.
  Our idea is to define the line $l_m$ so that it contains an intersection point 
  of $c_\ell$ and $c_r$, and it does not enter the interior of the grey region.
  A simple computation with sage~\cite{sagemath} outputs equations defining
  the two curves:
  \[c_\ell: 138x^2 - 568xy - 1071y^2 - 3018x + 8828y + 15312 = 0\enspace,\]
  \[c_r: 138x^2 - 156xy - 356y^2 - 1791x + 3296y + 1620 = 0\enspace.\]
  See Appendix~\ref{apx:CurveComputation} for the sage code for this
  computation.
  It can be checked, even by hand, that the point 
    \[p = (3.5 + 5\sqrt{2}, 1.5\sqrt{2}) \approx (10.57, 2.12)\]
    lies on both curves, and also
    on the line 
    $l_m = \{ \, (x,y)\, : \, y=0.3x-1.05 \, \}$.
    Therefore, $p$ is a feasible (and at the same time irrational) position for the guard $g_m$. Moreover, by
    plotting $c_\ell$, $c_r$, and $l_m$ in $\poly$ as in
    Figure~\ref{fig:FullPolygon}, we get an indication
    that as we traverse $l_m$ from left to
    right, at the point $p$
    we exit the area where $g_m$ and $g_l$ can guard together the two left pockets, and at the same time we enter the area
    where $g_m$ and $g_r$ can guard together the two right pockets.
    Thus, the only feasible position for the guard $g_m$ is the irrational point $p$.
    A formal proof will be given in Section~\ref{proofSec}.
  
    \medskip
  \noindent \textbf{Searching for the Polygon.}
  The simplicity of the ideas behind our construction
  does not reflect the difficulty of finding the exact coordinates for the polygon $\poly$.
  The reader might for instance presume
  that most other choices of horizontal pockets would work, if
  the line $l_m$ is changed accordingly. 
  However, this is not the case. 
  
  It is easy to construct the pockets so that the corresponding
  curves $c_\ell$ and $c_r$ intersect at some point~$p$.
  We expect $p$ to be an irrational point in general, since
  the curves $c_\ell$ and $c_r$ are defined by two second degree polynomials,
  as indicated above.
  In our construction, we need to force $g_m$ to be on a line $l_m$ containing $p$, but we can
  only force $g_m$ to be on a rational line.
  Hence, we require the existence of a rational line $l_m$ that contains $p$.
  
  As any two rational lines intersect in a rational point,
  there can be at most one rational line containing the irrational point $p$.
  Moreover, there exists a rational line containing $p$ if and only if
  $p=(r_1 + r_2 \alpha, r_3 + r_4 \alpha)$ for some $r_1,r_2,r_3,r_4 \in \mathbb{Q}$,
  where $\alpha \in \mathbb{R} \setminus \mathbb{Q}$ is an irrational number.
  The equation of the rational line containing $p$ is then
  $y=\frac{r_4}{r_2}\cdot x + (r_3 - r_1 \cdot \frac{r_4}{r_2})$.
  We say that this line \emph{supports} $p$.
Therefore, we should not hope that the intersection point of the curves $c_\ell$ and $c_r$ defined by arbitrarily chosen pockets will have a supporting line.
 Our main idea to overcome this problem
  has been to reverse-engineer the polygon, after having chosen the positions
  of the guards.
  We chose three irrational guards, all with supporting rational lines,
  and then defined the pockets so that $g_m$ automatically
  became the intersection point between the curves $c_\ell$ and $c_r$ associated with
  the pockets.

  We chose all three guards to have coordinates of the form
  $(r_1 + r_2 \sqrt 2, r_3 + r_4 \sqrt 2)$ for $r_1,r_2,r_3,r_4 \in \mathbb{Q}$.
Assume, for the ease of presentation, that we already know that we
can end up with a
polygon described as follows.
(In our initial attempts, our polygons were much less regular.)
  The polygon should consist of the
  rectangle $R=[0,20]\times[0,4]$ with some pockets added.
  We would like the pockets to extrude
  vertically from the horizontal edges of $R$
  such that the pockets meet $R$ along the segments
  $(4,0)(8,0)$, $(12,0)(16,0)$, $(4,4)(8,4)$, and $(12,4)(16,4)$,
  respectively.
  
  We now explain the technique for constructing the bottom pocket to the left which
  should extrude from $R$
  vertically downwards from the corners $(4,0)$ and $(8,0)$.
  We have to define the edge $e_b^\ell$, which is
  the bottom edge in the pocket.
  We want $p_b^\ell$ to be a point on $e_b^\ell$ such that
  $g_\ell$ can only see the part of $e_b^\ell$ from $p_b^\ell$ and to the right,
  whereas
  $g_m$ can only see the part of $e_b^\ell$ from $p_b^\ell$ and to the left.
  Therefore, we define $p_b^\ell$ to be the intersection point between
  the line containing $g_\ell$ and $(4,0)$, and the line containing
  $g_m$ and $(8,0)$. It follows that $p_b^\ell$ is of the
  form $(r_1 + r_2 \sqrt 2, r_3 + r_4 \sqrt 2)$ for some $r_1,r_2,r_3,r_4 \in \mathbb{Q}$.
  Hence, there is a unique rational line $l$
  supporting $p_b^\ell$, and
  $e_b^\ell$ must be a segment on $l$.
  We therefore need that both of the points $(4,0)$ and $(8,0)$
  are above $l$, since otherwise we do not get a meaningful polygon.
  However, this is not the case for arbitrary choices of the guards
  $g_\ell$ and $g_m$.
  The other pockets add similar restrictions to the positions of the guards.
  
  In the construction we had to take care of other issues as well. In particular, the line $l_m$ which supports the guard $g_m$ cannot enter the grey region between the two curves $c_\ell$ and $c_r$, as otherwise the position of $g_m$ would not be unique, and the guard could be moved to a rational point. Also, the three lines $l_\ell, l_m, l_r$ supporting the three guards $g_\ell, g_m, g_r$ cannot intersect within the polygon.
  
  We experimented with the construction in GeoGebra~\cite{gg2}, where we had the
  possibility to draw the lines supporting
  $p_{t}^\ell,p_{b}^\ell,p_{t}^r,p_{b}^r$ and see how they
  were changing in an intricate way as we changed the coordinates of
  the guards.
  For most choices of the guards and other parts of the polygon, we did not
  get meaningful results.
  The great advantage of GeoGebra is that we could continuously vary all parts of the
  polygon and play around with all parameters, thus gaining an intuitive understanding
  of various dependencies. After experimenting for a while, we were able to
  produce feasible examples and then find more appealing examples
  with simpler coordinates etc.
  In particular, it was important to us that many edges of the polygon
  are axis-parallel, so that we could easier derive from
  our example a rectilinear polygon with the same property, i.e., that the optimal guard set requires
  points with irrational coordinates.

\section{Proof of Theorems \ref{thm:3-irrational} and \ref{thm:4-3-bound}}\label{proofSec}

\textbf{Basic observations.}
Recall the construction of the polygon $\poly$ as defined in
Section~\ref{sec:polygon}, and consider a guard set of $\poly$ of cardinality at most $3$.
Let $l_\ell, l_m, l_r$ be, respectively, the restrictions of the following
lines to $\poly$:
$$x=2,\quad y=0.3x-1.05,\quad\text{and}\quad x=19.$$

As argued in Section \ref{sec:intuition}, the triangular pockets enforce a guard onto each of these lines.

\begin{lemma}\label{lem:enforce-lines}
Consider any guard set $S$ for $\poly$ consisting of at most $3$ guards. Then (i) $|S|=3$, and (ii) there is one guard on each of the lines 
$l_\ell, l_m, l_r$. 
\end{lemma}
\begin{proof}
% [Proof of Lemma~\ref{lem:enforce-lines}]
Each triangular pocket $T^\ell_t,T^\ell_b, T^m_t, T^m_b, T^r_t, T^r_b$ has one vertex which is not
on the basic rectangle $[0,20] \times [0,4]$.
For each triangular pocket, 
we consider the points in $\poly$
that can see that vertex.
These positions correspond to the areas pictured in yellow and orange
in Figure~\ref{fig:ForceGlobal}.

It is straightforward to check that the only positions of guards that can see two such vertices are on the segments $l_\ell, l_m, l_r$.
Since these segments are non-intersecting,
at least three guards are needed to see the whole polygon $\poly$. If there are three guards, then there must be one guard on
each of the segments $l_\ell, l_m, l_r$.
\end{proof}

Now, consider the intervals $i_1=[0.5,0.6]$ and $i_2=[1.7,1.8]$. Similarly as for the case of triangular pockets, we can show that rectangular pockets $R_\ell, R_m, R_r$ enforce a guard with an x-coordinate in $[10.5,10.6]$, and two remaining guards with y-coordinates in $i_1$ and $i_2$. 

\begin{lemma}\label{lem:enforce-segments}
Consider any guard set $S$ for $\poly$ consisting of $3$ guards. Then one of the guards has an x-coordinate in $[10.5,10.6]$. For the remaining two guards, one has a y-coordinate in $i_1$ and the other one in $i_2$.
\end{lemma}
\begin{proof}
From Lemma \ref{lem:enforce-lines}, there must be one guard $g_\ell$ on $l_\ell$, one guard $g_m$ on $l_m$, and the last guard $g_r$ on $l_r$.
Recall that the rectangular pockets are as follows $R_\ell$: $[-10,0]\times [1.7,1.8]$, $R_r$: $[20,30]\times [0.5,0.6]$,
and $R_m$:  $[10.5,10.6]\times [4,8]$. It is straightforward to check that none of the guards $g_\ell, g_r$ can see the two top vertices of the pocket $R_m$. Therefore, the middle guard $g_m$ has to see both these vertices and it must have an x-coordinate in $[10.5,10.6]$. 

Then, as $g_m \in l_m$, the y-coordinate of $g_m$ is in $[2.1,2.13]$. Therefore, $g_m$ cannot see any of the left vertices of $R_\ell$, or any of the right vertices of $R_r$. These four vertices must be seen by the guards $g_\ell$ and $g_r$.

As some guard must see the bottom-left corner of the pocket $R_\ell$, it must be placed at a height of at least $1.7$. Then, this guard cannot see any of the right vertices of $R_r$. Therefore, the last guard must see both right vertices of $R_r$, and its height must be within $i_1=[0.5,0.6]$. Then, this guard cannot see any left vertex of the pocket $R_\ell$, and the second guard must see both left vertices of the pocket, and its height must be within $i_2=[1.7,1.8]$. 
\end{proof}

\medskip
\noindent
\textbf{Dependencies between guard positions.}
Let $\{g_\ell, g_m, g_r\}$ be a guard set of $\poly$, with $g_\ell \in l_\ell, g_m \in l_m$, and $g_r \in l_r$. We will now analyze dependencies between the positions of the guards that are caused by the horizontal pockets of $\poly$.
Recall that the non-axis-parallel edges of these pockets are denoted by $e^\ell_t$, $e^r_t$, $e^\ell_b$, and $e^r_b$.

% We start with two technical lemmas that are needed for Lemma~\ref{lem:left-right-only}.

\begin{figure}[htbp]
 \centering
 \includegraphics{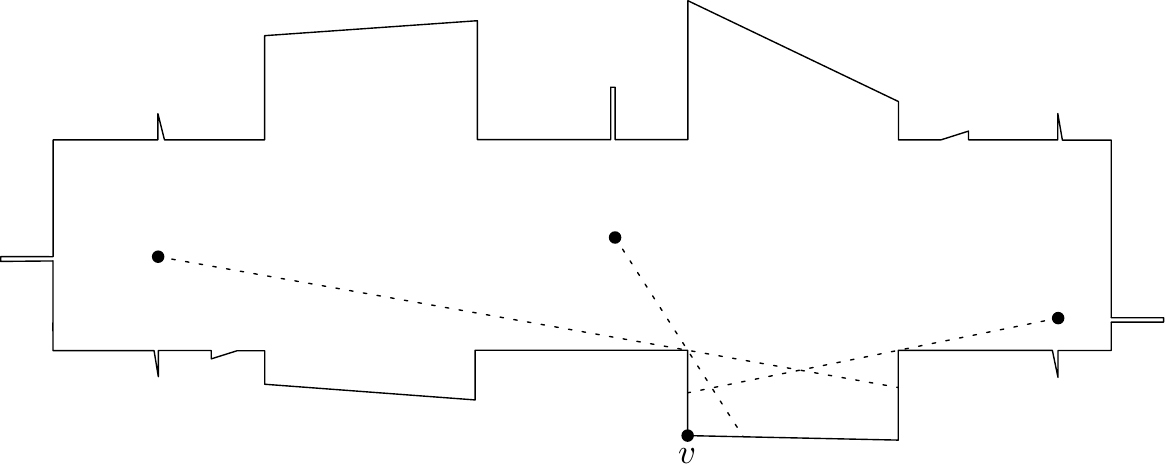}
 \caption{If the guard $g_r$ guards the wrong pocket, no guard can see the vertex $v$.}\label{fig:grWrongInterval}
\end{figure}

\begin{lemma}\label{lem:GuardIntervals}
   The $y$-coordinate of guard $g_\ell$ is in the interval $i_1 = [0.5 , 0.6]$ and the $y$-coordinate of guard $g_r$ is in the interval $i_2 = [1.7,1.8]$
\end{lemma}
\begin{proof}
  It is clear that one of the $y$-coordinate of $g_r$ is either in the interval $i_1$ or $i_2$ by the arguments given above. In case that the $y$-coordinate of $g_r$ is in $i_1$, it is easily seen that $g_r$ cannot see the bottom left corner $v$ of the bottom right point, see Figure~\ref{fig:grWrongInterval}.
  None of the other guards can see $v$ either --- a contradiction.
\end{proof}

\begin{restatable}{lem}{IndependentPockets}
  \label{lem:left-right-only}
    The guards $g_\ell$ and $g_m$ together see all of $e^\ell_t$ and $e^\ell_b$, and the guards $g_m$ and $g_r$ together can see all of $e^r_t$ and $e^r_b$.
  \end{restatable}

\begin{proof}
    By the construction of $\poly$, it holds that if a guard sees a point on one of the edges
    $e^\ell_t$, $e^r_t$, $e^\ell_b$, and $e^r_b$, then the guard sees an interval
    of the edge containing an endpoint of the edge. It now follows that if three guards together see
    one of these edges, then two do as well.
    Also note that it is impossible for a single guard to see either of the edges entirely.
    In order to prove the lemma, it thus suffices to prove that
\begin{itemize}
\item
    The part of $e^\ell_t$, $e^\ell_b$ that is seen by $g_r$ is also seen by $g_m$.
   
\item
    The part of $e^r_t$, $e^r_b$ that is seen by $g_\ell$ is also seen by $g_m$.
    \end{itemize}
   The Lemma follows from the two statements above. The two statements can be easily checked in Figure~\ref{fig:RightGuardObsolete}~and~\ref{fig:LeftGuardObsolete}.
\end{proof}

 \begin{figure}[t]
  \centering
   \begin{minipage}{.49\textwidth}
    \centering
    \includegraphics{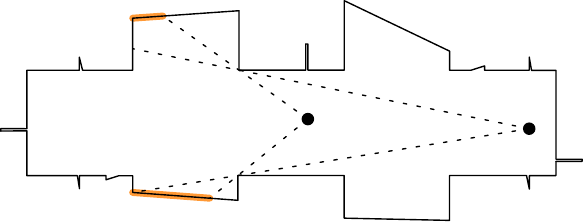}
    \subcaption{The part of $e^\ell_t$, $e^\ell_b$ that is seen by $g_r$ is also seen by $g_m$.}\label{fig:RightGuardObsolete}
  \end{minipage}
  \hfill
  \begin{minipage}{.49\textwidth}
    \centering
    \includegraphics{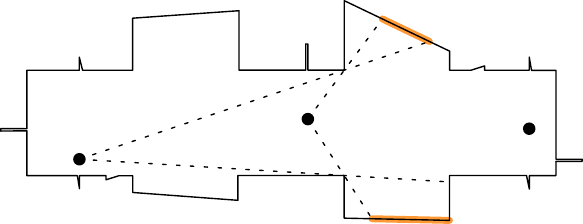}
    \subcaption{The part of $e^r_t$, $e^r_b$ that is seen by $g_\ell$ is also seen by $g_m$.}\label{fig:LeftGuardObsolete}
  \end{minipage}
  \caption{The left guard os not helpful to guard the right pockets and the right guard is not helpful to guard the left pockets.}
  \label{fig:ObsoleteGuards}
\end{figure}

\medskip
\noindent
\textbf{Computing the unique solution.} 
We can now show that there is only one guard set for $\poly$ consisting of three guards. Let us start by computing the right-most possible position of $g_m$ such that $g_\ell$ and $g_m$ can see together both left pockets.

For the next two lemmas, recall the notation from Figure~\ref{fig:FullPolygon}.

\begin{lemma}\label{lem:left-pockets}
The maximum $x$-coordinate of $g_m$ such that $g_\ell$ and $g_m$ can together see $e_t^\ell$ and $e_b^\ell$ is $x=3.5 + 5 \sqrt{2}$.
The corresponding position of $g_\ell$ is $(2,2-\sqrt{2})$.
\end{lemma}
\begin{proof}
Consider the guard $g_\ell$ at position $(2,h)$.
From Lemma~\ref{lem:GuardIntervals}, we know that $h\in i_1=[0.5,0.6]$.
We can easily compute \[p^\ell_t = \left(\frac{908-188h}{181-47h},\frac{7}{94} \cdot \frac{908-188h}{181-47h} + \frac{266}{47}\right)\]
If $g_m$ and $g_\ell$ together see $e_t^\ell$,
we know from Lemma~\ref{lem:left-right-only} that
$g_m$ has to be on or below the line  containing the vertices $(8,4)$ and $p_t^\ell$, i.e., the line with equation 
\[y=\frac{92-23 h}{-135 + 47 h} x + \frac{-1276+372 h}{-135 + 47 h}.\] 
As $g_m$ is on the line $l_m$ described by $y=0.3x-1.05$, its $x$-coordinate satisfies 
\[0.3x-1.05 \le \frac{92-23 h}{-135 + 47 h} x + \frac{-1276+372 h}{-135 + 47 h},\] i.e., \[x \le \frac{28355-8427 h}{2650 - 742 h}.\]

%From Lemma~\ref{lem:gl-visibility}
%in the appendix, if $h<\frac{9}{19} \approx 0.47$ or $h>\frac{135}{47} \approx 2.87$, then $g_m$ has to see one of the pockets entirely, i.e., its $x$-coordinate must be at most $8<3.5 + 5 \sqrt{2}$. Now, assume that $9/19 \le h \le 135/47$. 

If $g_m$ and $g_\ell$ together see $e_b^\ell$, then $g_m$ has to be on or above the line containing the vertices $(8,0)$ and 
\[p_b^\ell =\left(\frac{76h+12}{19h-3},-\frac{3}{38} \cdot \frac{76h+12}{19h-3} - \frac{6}{19}\right),\]
i.e., the line with equation 
\[y=\frac{3h}{19 h - 9} x - \frac{24h}{19 h - 9}.\] 
Hence, the $x$-coordinate of $g_\ell$ must satisfy 
\[0.3x-1.05 \ge \frac{3h}{19 h - 9} x - \frac{24h}{19 h - 9},\] 
i.e., \[x(1-h) \le \frac{81 h + 189}{54}.\]
Therefore, since $h<1$, we must have
\[x \le \frac{81 h + 189}{54 - 54h}.\] 

We now know that
\[x\leq
\min \left\{\frac{28355-8427 h}{2650 - 742 h},  \frac{81 h + 189}{54- 54 h}\right\}.\]
The first of the two values decreases with $h$, and the second one increases with $h$. Therefore the maximum is obtained when \[\frac{28355-8427 h}{2650 - 742 h} = \frac{81 h + 189}{54- 54 h},\] i.e., for $h=2-\sqrt{2}$. The value of $x$ is then $3.5 + 5 \sqrt{2}$. The corresponding position of the guard $g_\ell$ is $(2,h)=(2,2-\sqrt{2})$.
\end{proof}

Similarly, we can compute the left-most possible position of $g_m$ such that $g_m$ and $g_r$ can see together both right pockets. The proof is in the appendix.

\begin{lemma}\label{lem:right-pockets}
The minimum $x$-coordinate of $g_m$ such that $g_r$ and $g_m$ can see both
$e_t^r$ and $e_b^r$ is $x=3.5 + 5 \sqrt{2}$. The corresponding position of $g_r$ is $(19,1+\frac{\sqrt{2}}{2})$.
\end{lemma}

\begin{proof}
% [Proof of Lemma~\ref{lem:right-pockets}.]
Consider the guard $g_r$ at position $(19,h)$.
From Lemma~\ref{lem:GuardIntervals}, we know that $h\in i_2=[1.7,1.8]$.
If $g_m$ and $g_r$ together see $e_t^r$,
we know from Lemma~\ref{lem:left-right-only}
that
$g_m$ has to be on or below the line containing the vertices $(12,4)$ and 
\[ p^r_t = \left(\frac{4000 h - 9768}{250 h - 645}, -\frac{71}{150}\frac{4000 h - 9768}{250 h - 645}+\frac{4616}{375}\right),\]
i.e., the line with equation
\[y=\frac{46 h - 184}{250 h - 507} x + \frac{448 h + 180}{250 h - 507}.\] 
As $g_m$ is at the line  $y=0.3x-1.05$, its $x$ coordinate satisfies: 
\[0.3x-1.05 \le \frac{46 h - 184}{250 h - 507} x + \frac{448 h + 180}{250 h - 507},\] 
i.e., \[x \ge \frac{490 h - 243}{20 h + 22}.\]

If $g_m$ and $g_r$ together see $e_b^r$, then $g_m$ has to be on or above the line containing the vertices $(12,0)$ and
\[p_b^r = \left(\frac{224 h - 56}{14 h + 1}, -\frac{1}{42} \frac{224 h - 56}{14 h + 1}-\frac{4}{3}\right),\]  
i.e., the line with equation
\[y=\frac{6h}{17 - 14h} x - \frac{72h}{17 - 14h}.\] Hence, the $x$-coordinate of $g_r$ must satisfy 
\[0.3x-1.05 \ge \frac{6h}{17 - 14h} x - \frac{72h}{17 - 14h},\] 
i.e.,  $x \ge \frac{34 h - 7}{4 h -2}$. 

We have to minimize the value of \[\max \left\{\frac{490 h - 243}{20 h + 22}, \frac{34 h - 7}{4 h -2}\right\}.\] When the value of $h$ increases, the first of these two values increases, and the second one decreases. The minimum value is therefore obtained when 
\[\frac{490 h - 243}{20 h + 22} = \frac{34 h - 7}{4 h -2},\] i.e., for $h=1+\frac{\sqrt{2}}{2}$. The value of $x$ is then $3.5 + 5 \sqrt{2}$.
\end{proof}

We are now ready to prove our main theorems.

\begin{proof}[Proof of Theorem~\ref{thm:3-irrational}]
Let $\poly$ be the polygon constructed as in Section \ref{sec:polygon}, and let $S$ be a guard set for $\poly$ consisting of at most $3$ guards. From Lemma \ref{lem:enforce-lines} we have $|S|=3$, and there is one guard at each of the lines $l_\ell, l_m, l_r$. Denote these guards by  $g_\ell, g_m, g_r$, respectively.
From Lemma~\ref{lem:left-right-only} we know that if $g_\ell$, $g_m$, and $g_r$
together see all of $\poly$, then $g_\ell$ and $g_m$ must see all of
$e_t^\ell$ and $e_b^\ell$, and $g_m$ and $g_r$ must see all of
$e_t^r$ and $e_b^r$. It then follows from Lemmas~\ref{lem:left-pockets} and~\ref{lem:right-pockets} that $g_m$ must have coordinates
$(3.5 + 5 \sqrt{2}, 1.5 \sqrt{2})\approx (10.57,2.12)$,
%It is easy to check that the only point on the line
%$x=2$ that sees $e_t^\ell$ and $e_b^\ell$ together with $g_m$ is
$g_\ell=(2,2-\sqrt{2})\approx (2,0.59)$, and
%Likewise, the only point on $l_r$ that sees
%$e_t^r$ and $e_b^r$ together with $g_m$ is 
$g_r = (19,1+\frac{\sqrt{2}}{2})\approx (19,1.71)$.
Thus, indeed, the guards $g_\ell$, $g_m$, and $g_r$
see the entire polygon $\poly$ and are the only three guards doing so.

By scaling $\poly$ up by the least common multiple of the denominators in the
coordinates of the corners of $\poly$,
we obtain a polygon with integer coordinates. This does not affect
the number of guards required to see all of $\poly$.

In order to guard $\poly$ using four guards with rational coordinates, we
choose two rational guards $g'_{m,1}$ and $g'_{m,2}$ on $l_m$ a little
bit to the left and to the right of $g_m$, respectively.
The guard $g'_{m,1}$ sees a little more of both of the edges $e^\ell_t$
and $e^\ell_b$ than does $g_m$, whereas $g'_{m,2}$ sees a little more of
$e^r_t$ and $e^r_b$. Therefore, we can choose
a rational guard $g'_\ell$ on
$l_\ell$ close to $g_\ell$ such that
$g'_\ell$ and $g'_{m,1}$ together see
$e^\ell_t$ and $e^\ell_b$, and a rational
guard $g'_r$ on $l_r$ with analogous properties.
Thus, $g'_\ell,g'_{m,1},g'_{m,2},g'_r$ guard $\poly$.
\end{proof}

\begin{figure}[htbp]
\centering
\includegraphics[width=.8\textwidth]{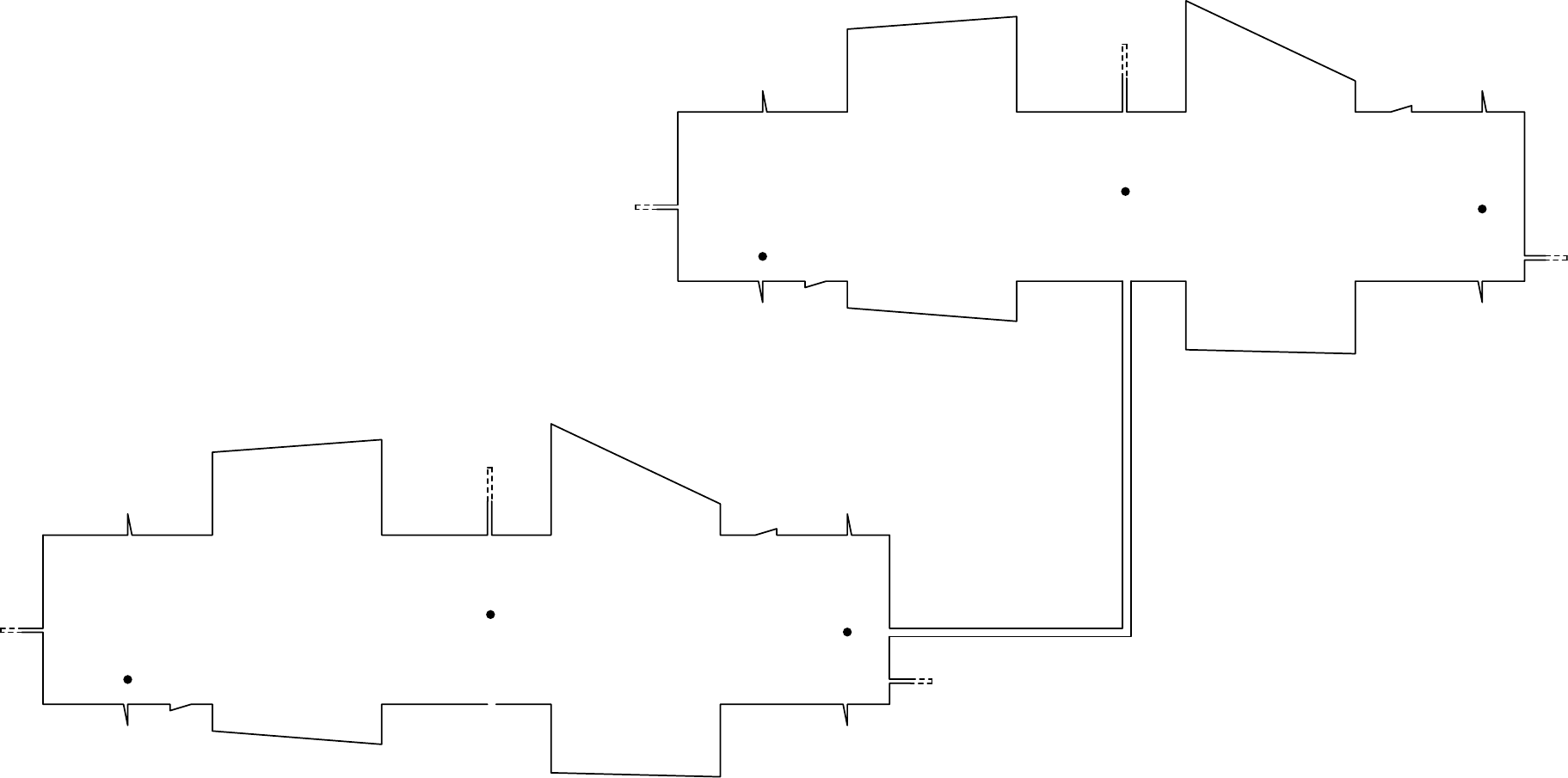}
\caption{A sketch of a polygon that can be guarded by $6$ guards when
irrational coordinates are allowed, but needs $8$ guards when
only rational coordinates are allowed.}
\label{fig43}
\end{figure}

\begin{proof}[Proof of Theorem~\ref{thm:4-3-bound}]
We will now construct a polygon $\poly_n$ that
can be guarded by $3n$ guards placed at points with irrational coordinates,
but such that when we restrict guard positions to points with rational coordinates,
the minimum number of guards becomes $4n$.
We start by making $n$ copies of the polygon $\poly$ described above,
which we denote by $\poly^{(1)},\ldots,\poly^{(n)}$.
We connect the copies into one polygon $\poly_n$ as follows.
Each consecutive pair $\poly^{(i)},\poly^{(i+1)}$
is connected by a thin corridor
consisting of a horizontal piece $H^{(i)}$
visible by the rightmost guard
in $\poly^{(i)}$, and a vertical piece $V^{(i)}$ visible to the middle
guard in $\poly^{(i+1)}$ (see Figure~\ref{fig43} for the case $n=2$).
%It follows that we need three guards in each $\poly^{(i)}$,
%and that they must all be placed on irrational coordinates.
We can then guard $\poly_n$ using $3n$ guards, by placing three guards within each polygon $\poly^{(i)}$ in the same way as for $\poly$, i.e., at
irrational points.

%To prove that $4n$ rational guards are needed, assume that we only
%use
Now, assume that $\poly_n$ can be guarded by at most $4n-1$ guards. We will show that at least one guard must be irrational.
For formal reasons,
we define $H^{(0)}=V^{(0)}=H^{(n)}=V^{(n)}=\emptyset$.
The horizontal and vertical corridors $H^{(i)}$ and $V^{(i)}$, for $i\in\{0,\ldots,n\}$,
intersect at a rectangular area $B^{(i)}=H^{(i)}\cap V^{(i)}$ which we call a \emph{bend}.
For $i\in\{1,\ldots,n-1\}$, the bend $B^{(i)}$ is non-empty and 
visible from both polygons $\poly^{(i)}$ and $\poly^{(i+1)}$.
Define the \emph{extension} of $\poly^{(i)}$, denoted by
$E(\poly^{(i)})$, to be the union of
$\poly^{(i)}$ and the adjacent corridors excluding the bends, i.e.,
$E(\poly^{(i)})=\poly^{(i)}\cup (V^{(i-1)}\setminus B^{(i-1)})
\cup (H^{(i)}\setminus B^{(i)})$.
Since the extensions are pairwise disjoint, there is an extension
$E(\poly^{(i)})$ containing at most three guards.
If there are no guards in any of the bends $B^{(i-1)},B^{(i)}$
it follows from Theorem~\ref{thm:3-irrational} that three guards must be placed
inside $\poly^{(i)}$ at irrational coordinates, so assume that there is
a guard in one or both of the bends.
If the adjacent corridors $V^{(i-1)}$ and $H^{(i)}$ are long enough and thin enough,
a guard in the bends $B^{(i-1)}$ and $B^{(i)}$ cannot see
any left corner of any of the vertical pockets of $\poly^{(i)}$,
any point in a triangular pocket, or any point in a horizontal pocket.
Hence, all the features of $\poly^{(i)}$ that enforce the irrationality of
the guards are unseen by the guards in the bends and it follows
that there must be irrational guards in $\poly^{(i)}$. Therefore, at least
$4n$ guards are needed if we require them to be rational. Similarly as in the proof of
Theorem \ref{thm:3-irrational}, we can show that $4n$ rational guards are enough to guard $\poly_n$.
\end{proof}

\section{Rectilinear Polygon}\label{sec:rectilinear}

  Figure~\ref{fig:RectilinearPolygonComplete2} depicts a rectilinear polygon
  $\poly_R$ with corners at rational coordinates that can be guarded by $9$ guards,
  but requires $10$ guards if we restrict the guards to points with rational
  coordinates.
  Before the formal proof, we want to give the reader a short overview. 
  The construction of $\poly_R$ starts with the polygon \poly from Theorem~\ref{thm:3-irrational}. We will extend the non-rectilinear parts by ``equivalent'' rectilinear parts,
  colored gray in the figure.
  The rectilinear pockets will be constructed in such a way, that each of them will require at least one guard in the interior. Additionally, if the interior of each pocket contains only one guard, then these guards must be placed at specific positions, making the
area not seen by these six additional guards %is 
exactly the polygon $\poly$
  described in Section \ref{sec:polygon} (the white area in Figure~\ref{fig:RectilinearPolygonComplete2}).
  Thus, the remaining $3$ guards must be placed at three irrational points
  by Theorem~\ref{thm:3-irrational}.
  
    \begin{figure}[htbp]
    \centering
    \includegraphics{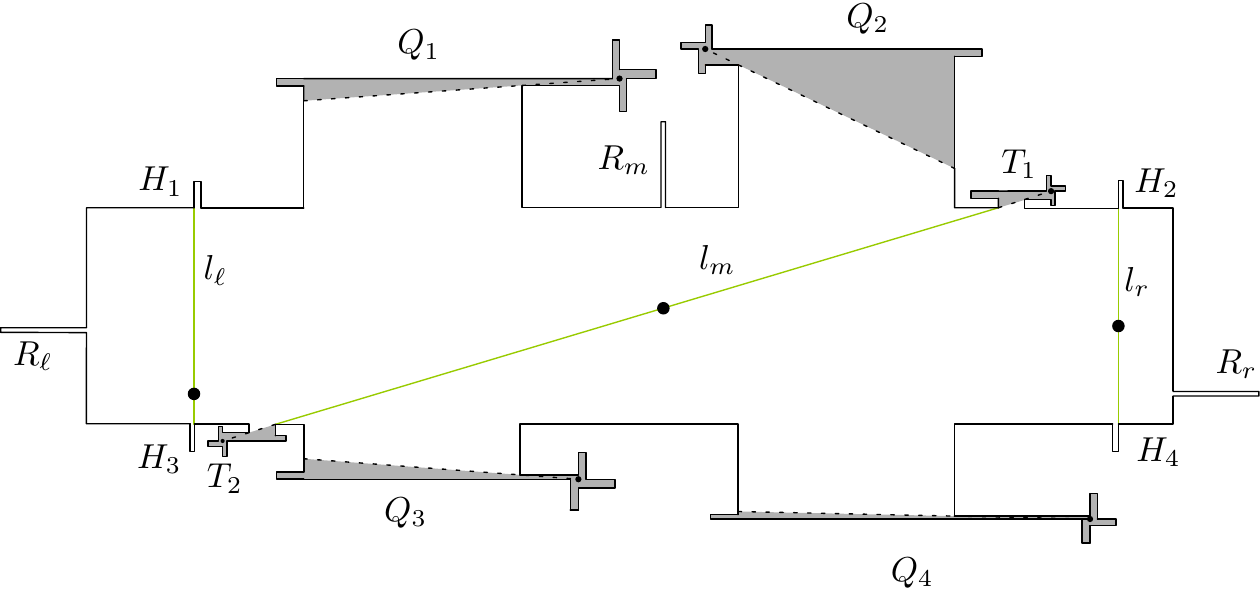}
    \caption{The rectilinear polygon $\poly_R$ can be guarded with $9$ guards only when we allow placing guards at irrational points.}
    \label{fig:RectilinearPolygonComplete2}
  \end{figure}

  \begin{proof}[Proof of Theorem~\ref{thm:rectilinear}.]
We describe a polygon $\poly_R$ with vertices at integer coordinates
that can be guarded by $9$ guards with irrational coordinates, but needs
$10$ guards if only rational coordinates are allowed.

  The construction of $\poly_R$ starts with the polygon \poly from Theorem~\ref{thm:3-irrational}. We will replace the non-rectilinear parts by ``equivalent'' rectilinear parts, see Figure~\ref{fig:RectilinearPolygonComplete2} for an illustration of the complete polygon $\poly_R$ and the notation therein. The additional areas need to be guarded by additional guards, as will be described later.
  
  First, consider the triangular pockets of the polygon $\poly$. These pockets have been added to enforce the guards to be on the lines $l_\ell,l_m,l_r$.
  Four of these pockets, the ones corresponding to $l_\ell$ and $l_r$, can be easily replaced by corresponding rectilinear pockets denoted by $H_1,H_2,H_3,H_4$, where three vertices of the new rectilinear pockets are the same as the vertices of the original triangular pockets.
  This does not work for the pockets corresponding to the line $l_m$, as this line is not axis-parallel, in particular, a guard on the line $l_m$ would not see all of
  the interior of such rectangular pockets.
  
  The two triangular pockets corresponding to $l_m$ and the four
  quadrilateral pockets will be extended to new, more complicated pockets.
  Note that there are only two different kinds of pockets that need to be extended, \emph{triangular pockets} and \emph{quadrilateral pockets}, as pictured on the left of Figure~\ref{fig:TrianglePocket} and Figure~\ref{fig:QuadranglePocket}.
  Each triangular pocket is defined by three vertices and one of the sides of each pocket is not axis-parallel.
  Similarly, each quadrilateral pocket is defined by four vertices and one of the sides of each pocket is not axis-parallel.
  
  Consider a pocket $P$, which needs to be extended in order to become
  rectilinear. Our extensions are pictured in the middle of Figure~\ref{fig:TrianglePocket} and~\ref{fig:QuadranglePocket}. 
  The green area in the middle of Figure~\ref{fig:TrianglePocket}  and~\ref{fig:QuadranglePocket} is a newly-created pocket $Q$.
  For now, let us assume that $Q$ does not intersect other parts of the polygon.
  The pocket $Q$ satisfies the following properties (see the right pictures in Figure~\ref{fig:TrianglePocket} and~\ref{fig:QuadranglePocket}).
  
 \begin{enumerate}[label=(\alph*), itemsep = 0.05cm]
  \item \label{itm:invisible} There are four points $p_1(Q),p_2(Q),p_3(Q),p_4(Q)$ within $Q$, such that each of them can only be seen by a guard, which is inside $Q$. 
  \item  \label{itm:Unique} There exists exactly one point $q(Q)$ that can see all four points 
  $p_1(Q),p_2(Q),p_3(Q),p_4(Q)$. 
  \item \label{itm:AllSeen} The point $q(Q)$ sees exactly the interior of $Q$.
  \item All vertices of $Q$ are rational.
 \end{enumerate}
  We now show that a pocket $Q$ satisfying all these properties can be constructed. First, we extend the non-axis-parallel edge of the pocket $P$ in the direction outside the
  polygon and place a point $q=q(Q)$, with rational coordinates, on it.
  We let $p_1,p_2,p_3,p_4$ be points with rational coordinates
  directly above, to the right, below,
  and to the left of $q$, respectively.
  Then, we construct four rectilinear
  sub-pockets each with a vertex at one of the points $p_1,p_2,p_3,p_4$,
  so that all these can be seen by $q$.
  These pockets can also be constructed with rational coordinates because $q$ has rational coordinates. Clearly, we can choose the point $q$ close enough to $P$ so that the resulting pocket $Q$ does not intersect the rest of the polygon.
  
  Let $\poly_R$ be the constructed rectilinear polygon as pictured in Figure~\ref{fig:RectilinearPolygonComplete2}, where all triangular and horizontal pockets have been extended by rectilinear pockets. We have to show that $\poly_R$ can be
  guarded by $9$ guards, but that we need $10$ guards if
  we require the guards to be at rational coordinates.
  The underlying idea is that after an optimal placement of one guard in each of the
  six pockets that have been extended in order to become rectilinear,
  the remaining area that must be seen by the remaining guards is exactly the same as in the original polygon $\poly$.
   
   \begin{figure}[htbp]
    \centering
    \includegraphics{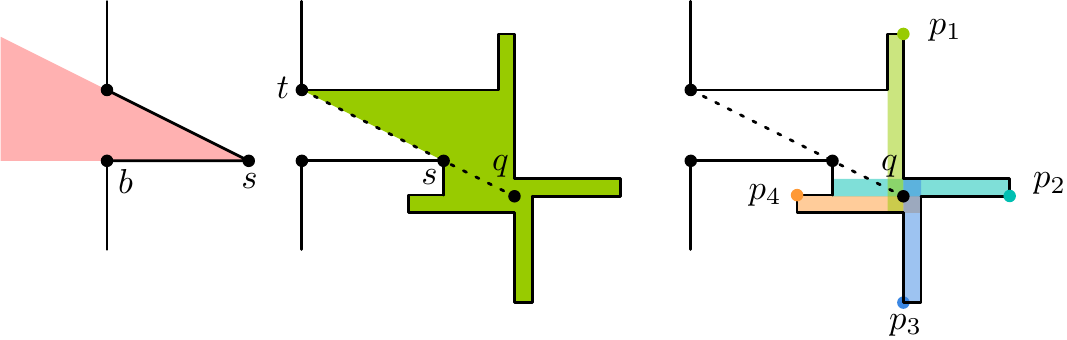}
    \caption{A triangular pocket is extended into a new rectilinear pocket.}
    \label{fig:TrianglePocket}
  \end{figure}

  \begin{figure}[htbp]
    \centering
    \includegraphics{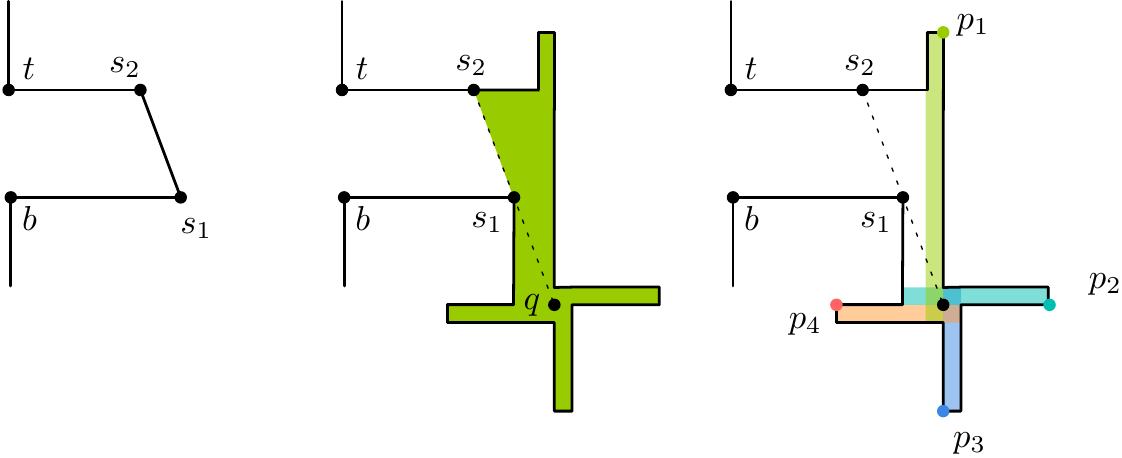}
    \caption{A quadrangular pocket is extended into a new rectilinear pocket.}
    \label{fig:QuadranglePocket}
  \end{figure}

  We first present a solution with $9$ guards
  when we are allowed to place guards at points with irrational coordinates.
  For this we place guards
  at the points $q(T_1)$, $q(T_2)$, $q(Q_1)$, $q(Q_2)$, $q(Q_3)$, $q(Q_4)$
  so that the interior of each of the pockets
  $T_1,T_2, Q_1,Q_2,Q_3,Q_4$ is seen,
  see Property~\ref{itm:Unique} and~\ref{itm:AllSeen}.
  Then we cover the remaining part of the polygon with three irrational guards
  as described in the proof of Theorem~\ref{thm:3-irrational}.
  
  It remains to show that $10$ guards are required when we
  restrict the guards to have rational coordinates.
  Suppose for the purpose of contradiction that there is a solution with $9$ rational guards.
  Note that there must be at least one guard in each pocket $Q\in \{T_1,T_2, Q_1,Q_2,Q_3,Q_4\}$ because of Property~\ref{itm:invisible}.
  We will now show that there must be at least three guards placed outside of $T_1 \cup T_2 \cup Q_1 \cup Q_2 \cup Q_3 \cup Q_4$.
  First, notice that no guard placed in any of the pockets $ T_1$, $T_2$, $Q_1$, $Q_2$, $Q_3$, $Q_4$ can see any of the following points: the top-left vertex of $H_1$ and $H_2$, and the bottom-right vertex of $H_3$ and $H_4$. To see these four points, at least two guards are needed. If there are only two guards, one of them must lie on $l_\ell$, and the other on $l_r$. But then none of the guards placed on $l_\ell \cup l_r \cup T_1 \cup T_2 \cup Q_1 \cup Q_2 \cup Q_3 \cup Q_4$ can see the top edge of the pocket $R_m$, and one more guard is needed. Therefore, at least three guards must be placed outside of $T_1 \cup T_2 \cup Q_1 \cup Q_2 \cup Q_3 \cup Q_4$.
  
 When only $9$ guards are available, there must be exactly $3$ guards outside the pockets $T_1$, $T_2$, $Q_1$, $Q_2$, $Q_3$, $Q_4$, and exactly one guard inside each pocket $Q \in \{T_1,T_2, Q_1,Q_2,Q_3,Q_4\}$.
  As each pocket $Q$ contains exactly one guard,
  then this guard must be the point $q(Q)$ because of Property~\ref{itm:invisible} and~\ref{itm:Unique}.
  Let $\poly_R^*$ be the area unseen by the guards within the six pockets.
  This polygon is exactly $\poly$, and by Theorem~\ref{thm:3-irrational} the unique solution with three guards is irrational. %We obtain property (ii).
  A polygon with integer coordinates can then be obtained by multiplying all coordinates with the least common multiple of all denominators of the coordinates.
  \end{proof}

%----------------------------------------------------------------------------------------------
\section{Future Work}\label{sec:future}
 One of the most prominent open questions related to the art gallery problem
 is whether the problem is in $\NP$.
Recently, some researchers popularized an interesting complexity class, called $\exists \R$, being somewhere between $\NP$ and $\text{PSPACE}$~\cite{Cardinal:2015:CGC:2852040.2852053, Schaefer2010, canny1988some, DBLP:journals/corr/Matousek14}.
  Many geometric problems for which membership in $\NP$ is uncertain
  have been shown to be complete for the complexity class $\exists \R$.
  Famous examples are:
  order type realizability,
  pseudoline stretchability, recognition of segment intersection graphs, 
  recognition of unit disk intersection graphs,
  recognition of  point visibility graphs,
  minimizing rectilinear crossing number,
  linkage realizability.
  This suggests that there might indeed be no polynomial sized witness for any of these problems as this would imply $\NP = \exists \R$.
It is an interesting open problem whether the art gallery problem is
$\exists\R$-complete or not.

The irrational coordinates of the guards in our examples are all of degree $2$, i.e.,
they are roots in second-degree
polynomials with integer coefficients. We would like to know if polygons exist where
irrational numbers of higher degree are needed in the coordinates of an optimal
solution.

We have constructed a simple polygon requiring three
guards placed at points with irrational coordinates.
It is a natural question whether there exists a polygon which can be guarded by
two guards only if they are placed at points with irrational coordinates.

We show that there exists polygons for which
$|OPT_{\Q}|\geq \frac{4}{3}|OPT|$.
It follows from the work by
Bonnet and Miltzow~\cite{DBLP:journals/corr/BonnetM16b} that
it always holds that
$|OPT_{\Q}|\leq 9 |OPT|$.
It is interesting to see if any of these bounds can be improved.

\medskip
\noindent
\textbf{Acknowledgement.} We want to thank  S\'{a}ndor Fekete, Frank Hoffmann, Udo Hoffmann, Linda Kleist, P\'{e}ter Kutas, G\"{u}nter Rote and Andrew Winslow for discussions on the problem and links to the literature.
Special thanks goes to Micha\l{} Adamaszek for providing the sage code.

  Mikkel Abrahamsen is partially supported by Mikkel Thorup's
  Advanced Grant from the Danish Council for Independent Research
  under the Sapere Aude research career programme.
  Anna Adamszek is supported by the Danish Council for Independent Research DFF-MOBILEX mobility grant.
  Tillmann Miltzow is supported by 
  the ERC grant ``PARAMTIGHT: Parameterized complexity and the search for tight complexity results", no. 280152.
   We want to further thank the developers of the software GeoGebra. 
   Being able to do computations and visualize parameter changes in real time facilitated our search tremendously.

\bibliographystyle{plain}
\bibliography{refs}

% \vfill
\newpage
\appendix

\section{Computations}\label{apx:CurveComputation}
\lstset{ %
  backgroundcolor=\color{white},   % choose the background color; you must add \usepackage{color} or \usepackage{xcolor}; should come as last argument
}

\begin{lstlisting}
def colinear(A,B,C):
    return Matrix([[A[0],A[1],1],[B[0],B[1],1],[C[0],C[1],1]]).determinant()
   
R.<t,p1,q1,p2,q2,x,y> = QQ[]

eq1 = ideal(
colinear((2,t),(4,4),(p1,q1)),
colinear((2,t),(4,0),(p2,q2)),
colinear((4,280/47),(p1,q1),(8,294/47)),
colinear((4,-12/19),(p2,q2),(8,-18/19)),
colinear((p1,q1),(8,4),(x,y)),
colinear((p2,q2),(8,0),(x,y))
    ).elimination_ideal([t,p1,q1,p2,q2]).gens()[0]

eq2 = ideal(
colinear((19,t),(16,4),(p1,q1)),
colinear((19,t),(16,0),(p2,q2)),
colinear((16,1776/375),(p1,q1),(12,2486/375)),
colinear((16,-36/21),(p2,q2),(12,-34/21)),
colinear((p1,q1),(12,4),(x,y)),
colinear((p2,q2),(12,0),(x,y))
    ).elimination_ideal([t,p1,q1,p2,q2]).gens()[0]

print eq1
print eq2
\end{lstlisting}

\end{document}